\documentclass[letterpaper]{article}
\PassOptionsToPackage{table}{xcolor}
\usepackage[preprint]{style}
\usepackage[hyphens]{url}
\usepackage{graphicx}
\usepackage{natbib}
\usepackage{caption}
\usepackage{algorithm}
\usepackage{algorithmic}

\usepackage{newfloat}
\usepackage{listings}
\DeclareCaptionStyle{ruled}{labelfont=normalfont,labelsep=colon,strut=off}
\floatstyle{ruled}
\newfloat{listing}{tb}{lst}{}
\floatname{listing}{Listing}

\usepackage{booktabs}
\usepackage{amsmath, amssymb, amsthm}
\usepackage{multirow}
\usepackage{xspace}
\usepackage{enumitem}
\usepackage{comment}
\newtheorem{theorem}{Theorem}

\newcommand{\ours}{\textsc{SpreadMark}\xspace}
\newcommand{\BitAcc}{BA\xspace}
\newcommand{\DetectAcc}{DA\xspace}
\definecolor{rowgray}{gray}{0.9}

\title{\ours: Robust Image Watermarking via Spread-Spectrum Embedding}

\author{
    Wei Song\textsuperscript{\rm 1}\equalcontrib,
    Yuxin Cao\textsuperscript{\rm 2}\equalcontrib,
    Zhenchang Xing\textsuperscript{\rm 3},
    Liming Zhu\textsuperscript{\rm 3},\\
    Jin Song Dong\textsuperscript{\rm 2},
    Yulei Sui\textsuperscript{\rm 1},
    Jingling Xue\textsuperscript{\rm 1}
}
\affiliations{
    \textsuperscript{\rm 1}University of New South Wales, Australia\\
    \textsuperscript{\rm 2}National University of Singapore, Singapore\\
    \textsuperscript{\rm 3}CSIRO's Data61, Australia
}

\begin{document}
\pagestyle{plain}
\maketitle

\begin{abstract}
Invisible image watermarks are increasingly used for deepfake detection and provenance tracking, where they must survive not only incidental distortions but also deliberate removal. We revisit spread-spectrum embedding, a classical watermarking principle, inside a modern neural post-hoc watermarking architecture. Our starting point is a measurement: in existing encoder--decoder schemes each message bit occupies only a small fraction of the image, a shared contributing factor to their fragility, since removal then need only disturb the region a bit occupies. \ours instead spreads each bit as a dense pseudo-random codeword over the whole image and recovers it by matched-filtering a learned cover-suppressed chip representation, with a parallel convolutional decoding path and sparsification-aware training. A conditional chip-space analysis shows that, under a codeword-independent perturbation model, dense spreading increases the budget required to disrupt matched-filter recovery. Evaluated on COCO and DIV2K against nine schemes, \ours is the only evaluated method retaining high detection under both the regeneration and the latent-space sparsification settings we test, with competitive JPEG and additive-noise robustness. It keeps the embedded watermark imperceptible, maintaining high perceptual quality on both COCO and DIV2K.
\end{abstract}

\section{Introduction}

The deployment of generative models for photorealistic image synthesis has made invisible watermarking a frontline tool for provenance and deepfake mitigation \cite{fernandez2023stable, yang2024gaussian, wen2023treering, bui2023trustmark}. The dominant paradigm is the encoder--decoder watermarking scheme, in which a learned encoder embeds a message $b$ into a cover image $x$ to produce a visually identical $x_w = E(x,b)$, and a learned decoder recovers $b$ from $x_w$ or a distorted derivative. Post-hoc schemes~\cite{zhu2018hidden, jia2021mbrs, ma2022cin, wu2023sepmark, bui2023trustmark, zhang2019rivagan, sander2025wam} all follow this template, and many generation-time schemes~\cite{yang2024gaussian, ci2024ringid, lu2024vine} fold the same encoder--decoder design into the latent space of a generative model.

For provenance, a watermark must survive not only incidental distortions such as JPEG compression, noise, blur, and cropping, but also deliberate erasure, now both practical and varied: adversarial perturbations optimize against the decoder, regeneration attacks rewrite the image through a generative model~\cite{an2024waves, ni2025diffusionedit}, and latent-space sparsification suppresses the watermark in a query-free, black-box manner~\cite{songdemark2026}. A watermark offered as evidence of origin is only as trustworthy as its weakest point under this combined threat.

One factor is shared across these attacks: concentrated residual embedding. Trained only to be imperceptible and decodable, the encoder is free to place each bit in a small part of the image, and we measure that it does (Figure~\ref{fig:prelim}). Such a footprint is a plausible weak point for every removal attack we consider: a distortion corrupts the local residual, an adversarial step suppresses the decoder-sensitive features, a regeneration pass rewrites the non-semantic trace, and a latent-space sparsification removes the compact subspace, each by disturbing only the small region the bit occupies. We take this as a motivating structural hypothesis, not a demonstrated cause; it suggests distributing each bit as widely across the image as the imperceptibility budget allows.

\begin{figure*}[t]
\centering
\includegraphics[width=\linewidth]{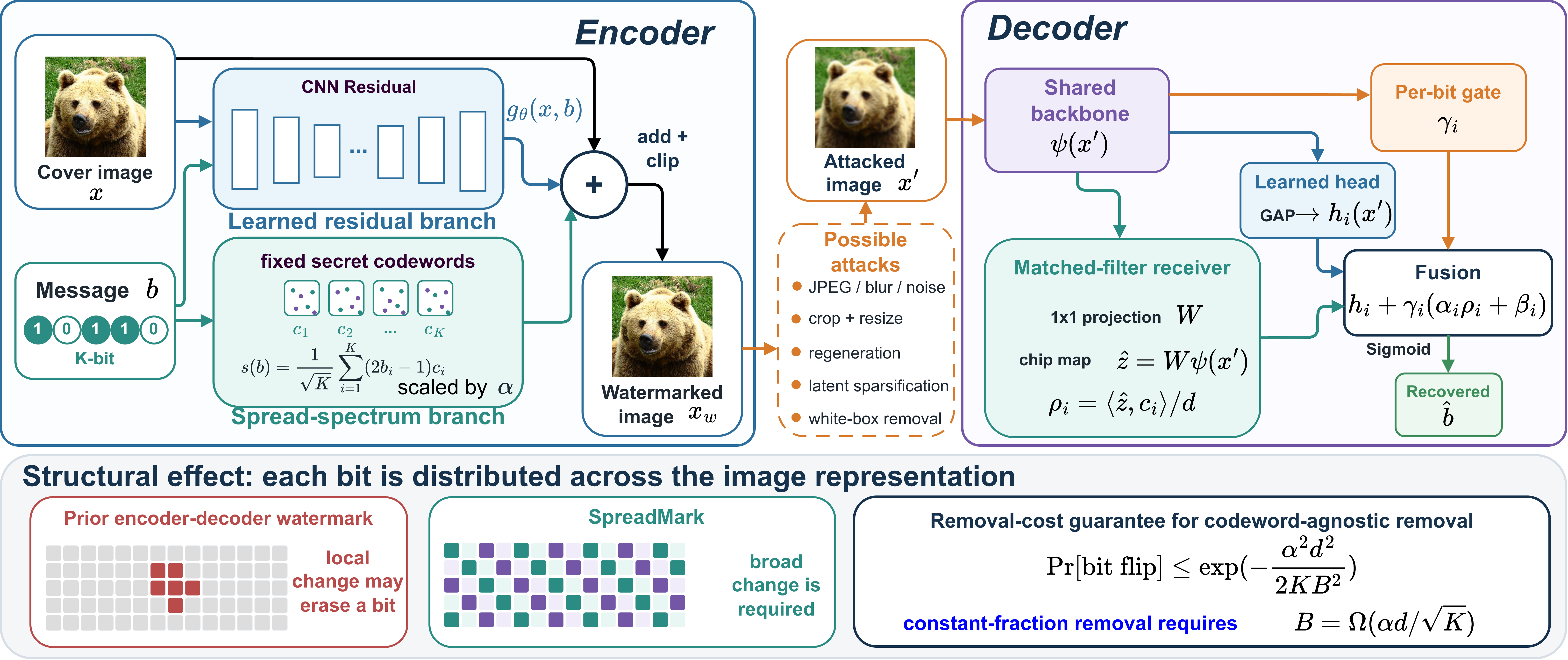}
\caption{Overview of \ours. The encoder adds to cover image $x$ a learned CNN residual $g_\theta(x,b)$ and a spread-spectrum signal $s(b)$ superposing one secret codeword per bit. From a possibly attacked $x'$, the decoder fuses a learned head $h_i$ with a per-bit-gated matched-filter statistic $\rho_i = \langle\hat z,c_i\rangle/d$ on a cover-suppressed chip $\hat z$. Bottom: under the codeword-independent perturbation model of Theorem~\ref{thm:bound}, broad spreading raises the removal budget to $B = \Omega(\alpha d/\sqrt K)$.}
\label{fig:arch}
\end{figure*}

Spreading a message over many dimensions is a classical idea, not a new one: spread-spectrum communication distributes a low-rate message across a wide band through pseudo-random codes and recovers it by correlation~\citep{proakis2007digital, viterbi1995cdma}, and classical image watermarking applied the same principle over transform-domain coefficients~\citep{cox1997spread}. \ours brings this principle into a neural post-hoc watermarking architecture: each bit is assigned a fixed pseudo-random codeword spanning the image, the message is embedded as the superposition of the selected codewords on top of a learned residual, and the decoder correlates a learned cover-suppressed chip representation against the same codeword bank, retaining a convolutional path as a fallback where a fixed spatial code desynchronizes. A sparsification-aware training stage further hardens the learned decoder against low-rank feature suppression.

We analyze the resulting receiver in chip space. When the perturbation an attack induces there is statistically independent of the secret codebook, the attacker cannot align it with any codeword, and the budget needed to disrupt matched-filter recovery grows with how widely each bit is spread. Empirically, this dense-spreading design yields strong robustness under regeneration and latent-space sparsification and remains comparable to the best baselines under JPEG compression and additive noise.

In summary, we make four contributions:

\begin{itemize}[leftmargin=*, topsep=2pt, itemsep=0pt, parsep=0pt]
\item We measure the image-space per-bit footprint of nine schemes and find it uniformly small, and we advance this as a motivating hypothesis for their fragility rather than a demonstrated cause of it.
\item We propose \ours, combining pseudo-random spread-spectrum injection, sparsification-aware training, and a learned cover-suppressed chip decoded by matched-filter and convolutional paths.
\item We give a conditional chip-space guarantee: under a codeword-independent perturbation model, dense spreading increases the budget an attacker needs to disrupt recovery by the matched filter.
\item We evaluate \ours on COCO and DIV2K against nine schemes under four attack families, trace its behavior through ablations, and state plainly where it fails.
\end{itemize}

\section{Related Work}

\noindent
\textbf{Image Watermarking.}
Image watermarking has developed along two directions. HiDDeN~\citep{zhu2018hidden} established the post-hoc encoder--decoder template, embedding a message and recovering it after a differentiable noise layer; later schemes extend it with mixed real and simulated JPEG~\citep{jia2021mbrs}, attention-based encoding~\citep{zhang2019rivagan}, invertible networks~\citep{ma2022cin}, separable robust and fragile components~\citep{wu2023sepmark}, universal pretrained systems~\citep{bui2023trustmark, sander2025wam}, edit localization~\citep{zhang2024editguard}, provable robustness~\citep{xian2024raw}, higher capacity~\citep{evennou2025latentseal}, and decoupled defenses~\citep{chen2026advmark}. Generation-time schemes instead embed during synthesis, by fine-tuning the latent decoder~\citep{fernandez2023stable}, structuring the initial or sampling noise~\citep{wen2023treering, yang2024gaussian, ci2024ringid}, injecting marks into latents~\citep{zhang2024zodiac}, learning a diffusion prior~\citep{lu2024vine}, or binding the mark to semantics~\citep{arabi2025seal, zhang2026sembind}, the last being vulnerable to coherence-preserving removal~\citep{gao2026llmcoherence}. \ours belongs to the first direction: it watermarks an existing image, independent of how it was produced, so we compare against post-hoc schemes and treat generation-time watermarking as a complementary line that requires control of the synthesis process, and hence a different deployment model from the post-hoc setting we study.

\smallskip
\noindent
\textbf{Watermarking Attacks.}
We consider four families: standard distortions; white-box adversarial removal, which optimizes a perturbation against the decoder by projected gradient descent; regeneration, which rewrites the image through a generative model while preserving content and can substantially weaken a watermark in a single pass~\citep{an2024waves, ni2025diffusionedit}; and latent-space sparsification, which suppresses it through low-rank feature projections in a query-free, black-box setting~\citep{songdemark2026}. None of the nine schemes we evaluate withstands all four families.

\smallskip
\noindent
\textbf{Spread-Spectrum Watermarking.}
Spreading a message across many carriers is a classical principle, not a contribution of this work. \citet{cox1997spread} brought spread spectrum~\citep{proakis2007digital, viterbi1995cdma} to multimedia, embedding an i.i.d.\ Gaussian carrier in perceptually significant DCT coefficients and detecting it by normalized correlation; improved spread spectrum cancels host interference at the correlator~\citep{malvar2003improved}, and the construction was carried into video~\citep{hartung1998watermarking}. Learned schemes have also distributed the mark spatially: \citet{plata2020robust} train a CNN encoder to spread the message over the spatial domain and lower the local bit density, \citet{fernandez2022watermarking} modulate a message into the signs of projections onto random orthogonal carriers in a frozen self-supervised feature space, and dispersed watermarking replicates the message across randomly placed blocks and re-synchronizes them at extraction~\citep{guo2023practical}. \ours is therefore not the first scheme to distribute watermark information spatially. Its novelty is the combination: an explicit pseudo-random codeword per bit injected in the image domain alongside a learned residual, a learned cover-suppressed chip representation on which the matched filter operates, parallel matched-filter and convolutional decoding paths under a per-bit gate, sparsification-aware training, and a systematic evaluation against regeneration and latent-space removal. Relative to the classical transform-domain constructions the carrier is dense in the pixel basis and the correlator input is learned rather than fixed; relative to prior neural spatial-spreading schemes the code is explicit and shared by encoder and decoder rather than implicit in the encoder, which is what makes the matched-filter path separately ablatable and analyzable, unlike an implicit code.

\section{Preliminary Analysis}
\label{sec:prelim}

\begin{figure}[t]
\centering
\includegraphics[width=\columnwidth]{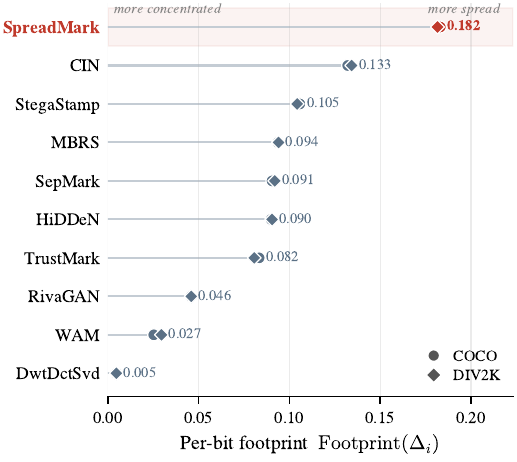}
\caption{Image-space per-bit footprint on COCO and DIV2K. For each scheme we flip one bit, forming $\Delta_i$ (Eq.~\ref{eq:delta}), and compute $\mathrm{Footprint}(\Delta_i)$ (Eq.~\ref{eq:footprint}), the effective fraction of image-space coordinates it occupies.}
\label{fig:prelim}
\end{figure}

We first measure, in image space on the watermarked output, how widely existing watermarks distribute a single message bit. We study the nine schemes listed in Table~\ref{tab:main}, using 300 watermarked images per dataset from COCO and DIV2K. For a fixed image and watermark, we flip the $i$-th bit and compute the resulting change vector:
\begin{equation}
\label{eq:delta}
\Delta_i =
\mathrm{vec} \left(
u(x,b^{(i=1)}) - u(x,b^{(i=0)})
\right)
\in \mathbb{R}^{d},
\end{equation}
where $u(\cdot)$ is the watermarked output image and $\mathrm{vec}(\cdot)$ flattens it. The per-bit footprint is the normalized participation ratio of this change vector:
\begin{equation}
\label{eq:footprint}
\mathrm{Footprint}(\Delta_i)
=
\frac{1}{d}
\cdot
\frac{
\left(\sum_{j=1}^{d} \Delta_{i,j}^{2}\right)^2
}{
\sum_{j=1}^{d} \Delta_{i,j}^{4}
}.
\end{equation}
This estimates the effective fraction of image-space coordinates carrying the bit-induced change: it approaches 1 for a uniform change and is approximately $m/d$ when the same energy is concentrated in $m$ coordinates, so it measures how far a bit is spread independently of how strongly it is embedded.

Figure~\ref{fig:prelim} shows that prior schemes have consistently small image-space per-bit footprints and that \ours yields a substantially larger one. Two qualifications apply. The measurement is correlational: a larger footprint is consistent with, but does not establish, greater removal robustness, and these schemes differ in architecture, payload, and embedding strength as well as in footprint. It is also specific to image space; an image-space-dense watermark may still be concentrated once an encoder maps it into a learned latent, which is where latent-space attacks act. We therefore read Figure~\ref{fig:prelim} as motivating the design principle---spread each bit as widely as the imperceptibility budget allows---and test it empirically.

\section{Method}
\label{sec:method}

Figure~\ref{fig:arch} gives an overview of \ours. The encoder superposes an explicit spread-spectrum signal on a standard additive residual, spreading each message bit across the whole image as a fixed pseudo-random codeword. The decoder recovers it by matched-filtering a cover-suppressed chip representation against the shared codeword bank, backed by a convolutional head and a per-bit gate.

\subsection{Threat Model}
An encoder $E:(x,b)\mapsto x_w$ embeds a $K$-bit message $b$ into a cover image $x$, and a decoder $D:x'\mapsto\hat b$ estimates it from any, possibly attacked, image $x'$. The defender wants $x_w$ visually close to $x$ and $\hat b$ accurate under attack; the attacker receives $x_w$ and outputs $x_a$, reducing message recovery while preserving visual quality. Our four attack types are standard distortions and regeneration, both oblivious to the watermark; latent-space sparsification, query-free and black-box, with no access to the codewords or decoder; and adversarial removal, white-box, with full gradient access to the decoder and hence the codewords. The codebook is assumed secret throughout and is never rotated, an assumption whose consequences the Discussion sets out in full.

\subsection{Spread-Spectrum Embedding}
Prior schemes place most of each bit's energy in a low-dimensional subspace of the $d = HW$ image-space directions (Figure~\ref{fig:prelim}). In contrast, \ours adopts the additive-residual design of HiDDeN \cite{zhu2018hidden}, $x_w = \Pi\big(x + g_\theta(x,b)\big)$, with $g_\theta$ a small CNN and $\Pi$ a clip to the valid range, and augments it with a parallel spread-spectrum injection. We pre-generate $K$ fixed pseudo-random codewords $\{c_i\}_{i=1}^{K}$, each a single-channel map $c_i \in \mathbb{R}^{H\times W}$ with zero mean and unit per-pixel variance ($\tfrac1d\|c_i\|_2^2 = 1$); being i.i.d.\ Gaussian, distinct codewords are near-orthogonal, $\langle c_i,c_j\rangle/d \approx 0$ for $i \neq j$. The encoder embeds the message as:
\begin{equation}
\label{eq:embed}
x_w \;=\; \Pi \Big(\, x \;+\; g_\theta(x,b)
  \;+\; \alpha\,\underbrace{\tfrac{1}{\sqrt{K}}\textstyle\sum_{i=1}^{K}(2b_i - 1)\,c_i}_{s(b)} \,\Big),
\end{equation}
with $s(b)$ broadcast across RGB channels. The $1/\sqrt{K}$ factor keeps $s(b)$ at unit per-pixel variance independent of payload, so the learned gain $\alpha$, a softplus of a free parameter initialized small, is the per-pixel amplitude of the spread component and sets the bit-recovery and visual-quality trade-off.

The decoder $D$ couples a matched-filter receiver with a learned convolutional head. Given a possibly attacked image $x'$, a fully-convolutional backbone yields features $\psi(x')$; a $1 \times 1$ projection $W$ maps them to a single-channel, cover-suppressed ``chip'' map $\hat z = W\psi(x')$, correlated against the same codeword bank to form the matched-filter statistic $\rho_i = \langle \hat z, c_i\rangle / d$. The unit-variance normalization makes $\langle c_i,c_i\rangle/d = 1$, so an ideal chip yields $\rho_i \approx \alpha(2b_i - 1)/\sqrt{K}$, the sign of bit $i$. A pooled head $h(x')=\mathrm{Lin}_h(\mathrm{GAP}(\psi(x')))$, with $\mathrm{GAP}$ global average pooling \citep{lin2014nin} and $\mathrm{Lin}$ a linear map to the $K$ bits \citep{goodfellow2016deep}, supplies a parallel convolutional decoder, and a per-bit gate $\gamma_i=\sigma(\mathrm{Lin}_\gamma(\mathrm{GAP}(\psi(x'))))_i$ weights the matched-filter term:
\begin{equation}
\label{eq:decode}
\hat b_i \;=\; \sigma \big(\,h_i(x') + \gamma_i\,(a_i\,\rho_i + \beta_i)\,\big),
\qquad i=1,\dots,K,
\end{equation}
with learnable per-bit gain $a_i$ and bias $\beta_i$. The $\rho_i$ term is a correlation receiver against the dense code and dominates recovery where the fixed spatial code is preserved, as under JPEG compression and additive noise; the head $h$, supervised by an auxiliary bit loss so that it remains a competent decoder on its own, supplies robustness under geometric desynchronization and is the more resilient path under white-box adversarial removal, which targets the correlation receiver directly. The crucial property of (\ref{eq:embed})--(\ref{eq:decode}) is that each bit's contribution is dense in the pixel basis, with non-zero projection on all $d$ image-space directions, so a perturbation not aligned with the code must disturb a broad fraction of the image in order to erase it, as the next subsection makes precise.

\subsection{Sparsification-Aware Adversarial Training}

Of the four attack types, only standard distortions are addressed by training augmentation, which with probability $p_d$ replaces $x_w$ by one operation drawn from JPEG compression, Gaussian blur, additive noise, brightness, or crop-resize. Regeneration and white-box removal are never simulated. A sparsification-aware adversarial training (SA-AT) stage does simulate the sparsification attack; empirically it sharpens regeneration and JPEG robustness rather than sparsification, which the model attains without it (Table~\ref{tab:saat}).

The attack drives the image's deep features into a rank-$r$ subspace while staying close to $x_w$:
\begin{equation}
\label{eq:attack}
x_a \;=\; \arg\min_{\|\delta\|_\infty\le\varepsilon}
  \big\| (I - U_rU_r^\top)\,f(x_w + \delta) \big\|_2^2,
\end{equation}
where $f$ is a frozen feature extractor and $U_r$ collects the top-$r$ singular vectors of its activations on a clean batch; this is the operator used by latent-space sparsification attacks~\cite{songdemark2026}. With probability $p_s$ we simulate it during training: sample $r \sim \mathcal{U}\{r_{\min},r_{\max}\}$, take the rank-$r$ basis $U_r$, and run $T$ pixel-space gradient steps of (\ref{eq:attack}), recomputing the decomposition only every 200 steps. The surrogate used at evaluation time is built independently of the training one, so robustness does not come from matching a specific attacker.

Writing $\ell(z)=\mathrm{BCE}(b,z)$ for the bit loss against the true message and $\mathcal{A}$ for the composition of the two perturbations, each applied independently with its own probability and $\mathcal{A}=\mathrm{id}$ if neither fires, we train \ours with:
\begin{align}
\mathcal{L} = &\;\lambda_a\,\ell(D(\mathcal{A}(x_w))) + \lambda_b\,\ell(D(x_w)) + \lambda_h\,\ell(h(\mathcal{A}(x_w))) \nonumber\\
& + \lambda_v\big[\|x-x_w\|^2 + \mathrm{LPIPS}(x,x_w)\big],
\end{align}
where $\lambda_a$ is the robustness objective on the perturbed image, $\lambda_b$ preserves clean accuracy, $\lambda_h$ supervises $h$, and $\lambda_v$ controls visual quality through a combined $L_2$ and LPIPS penalty.

\subsection{A Conditional Chip-Space Guarantee}

We now make the design intuition precise for an idealized receiver, and state what the model does and does not cover. The codewords act as a secret key shared by encoder and decoder. The decoder reads bit $i$ from $\rho_i=\langle\hat z,c_i\rangle/d$ on the chip $\hat z=\alpha\,s(b)+\nu$, where $s(b)=\tfrac{1}{\sqrt K}\sum_j(2b_j-1)c_j$ and the codewords obey $\tfrac1d\|c_i\|_2^2=1$, $\langle c_i,c_j\rangle/d\approx 0$. The term $\nu$ collects the cover-suppression residual and cross-code interference; we idealize it as carrying no systematic component along any secret codeword, so the clean statistic concentrates at margin $|\mathbb{E}\,\rho_i|=\alpha/\sqrt K$ with sign $2b_i-1$.

\begin{theorem}[Chip-space budget under a codeword-independent perturbation]
\label{thm:bound}
Let an attack map the chip $\hat z\mapsto\hat z+\Delta$, and assume the induced perturbation $\Delta$ is statistically independent of the secret codebook $\{c_i\}$ and satisfies $\|\Delta\|_2\le B$. Then the induced shift $\Delta\rho_i=\langle\Delta,c_i\rangle/d$ satisfies $\mathbb{E}[\Delta\rho_i]=0$ and $\mathrm{Var}[\Delta\rho_i]\le B^2/d^2$, and
\begin{equation*}
\Pr \big[\widehat b_i\neq b_i\big]\;\le\;\exp \Big(-\tfrac{\alpha^2 d^2}{2KB^2}\Big).
\end{equation*}
Consequently, any such attack whose expected number of flipped bits reaches a constant fraction $\rho_0$ of the $K$ bits must use $B\ge\alpha d/\sqrt{2K\ln(1/\rho_0)}=\Omega(\alpha d/\sqrt K)$; in pixel space, with an $L$-Lipschitz chip map and $\|\delta\|_\infty\le\varepsilon$, this reads $\varepsilon=\Omega(\alpha\sqrt d/(L\sqrt K))$.
\end{theorem}

\begin{proof}[Proof Sketch]
Write $\Delta\rho_i=\tfrac1d\sum_j\Delta_j c_{i,j}$ with $c_{i,j}$ zero-mean, unit-variance, independent of $\Delta$. Then $\mathbb{E}[\Delta\rho_i]=0$, $\mathrm{Var}[\Delta\rho_i]\le B^2/d^2$, and $\Delta\rho_i$ is sub-Gaussian with this variance proxy~\citep{vershynin2018high}. A bit flips only if $\Delta\rho_i$ cancels the margin $\alpha/\sqrt K$, with probability at most $\exp(-\alpha^2 d^2/(2KB^2))$. Summing over $i$ gives the expected flip count, and requiring it to reach $\rho_0K$ yields the stated budget; $\|\Delta\|_2\le L\sqrt d\,\varepsilon$ converts chip budget to pixel budget. The Appendix gives the full argument.
\end{proof}

\noindent
\textbf{Scope.} The theorem is conditional and idealized, not a robustness proof for the deployed system: dense spreading raises the budget required to disrupt matched-filter recovery, and that budget falls by roughly $\sqrt\phi$ when a bit occupies only a fraction $\phi$ of the coordinates. Three restrictions matter. The independence assumption does not hold automatically against an attack that observes $x_w$, since $x_w$ itself depends on the codewords. The bound is a first-moment statement, not a high-probability one. And it covers neither adaptive attackers who estimate the code, hold the key, or optimize against the full decoder, nor regeneration, compression, denoising, or sparsification, which we treat empirically below.

\section{Experiments}
\label{sec:exp}

\subsection{Experiment Setup}
\label{subsec:experiment_setup}

\noindent
\textbf{Datasets.} We train on 4,500 images from the COCO~\citep{lin2014microsoft} training split and evaluate on 500 held-out COCO images and 500 images from the DIV2K~\citep{agustsson2017div2k} validation set, none seen in training.

\smallskip
\noindent
\textbf{Baselines.} We compare \ours against nine schemes, all run from their public implementations, in two groups. HiDDeN~\cite{zhu2018hidden}, MBRS~\cite{jia2021mbrs}, StegaStamp~\cite{tancik2019stegastamp}, and SepMark~\cite{wu2023sepmark} are re-trained in our pipeline under conditions matched to \ours: same training images, $128 \times 128$ resolution, $K = 30$ bits, optimizer, schedule, and distortion augmentation. CIN~\cite{ma2022cin}, TrustMark~\cite{bui2023trustmark}, RivaGAN~\cite{zhang2019rivagan}, WAM~\cite{sander2025wam}, and DwtDctSvd~\cite{barni1998dctwm} run from the authors' released checkpoints, which we cannot re-train under our protocol, at their native resolution (128 for CIN, 256 for the rest) and native payload (30 to 64 bits). Each scheme receives images bilinearly resized to its native resolution and every attack is applied there, so a spatial attack such as an 80\% crop is geometrically equivalent across schemes. The two groups sit at very different quality operating points (Table~\ref{tab:quality}), so robustness and imperceptibility must be read jointly rather than as two independent rankings of the same schemes.

\smallskip
\noindent
\textbf{Attacks.} \textit{Standard distortions:} JPEG at quality 75, additive Gaussian noise of standard deviation 0.05 on the $[-1,1]$ pixel scale, Gaussian blur of kernel standard deviation 2.0, and an 80\% center-crop followed by resizing. \textit{Adversarial removal:} a white-box projected-gradient attack on the decoder, 30 steps at step size 0.02, budget $\varepsilon = 0.10$, with an $\varepsilon \in [0.02,0.10]$ sweep in the Appendix (Figure~\ref{fig:eps}). \textit{Regeneration:} WAVES~\citep{an2024waves} img2img with Stable Diffusion v1.4~\citep{rombach2022sd} at noise strength 0.25, run at $512 \times 512$ and resized back to each scheme's native resolution. \textit{Latent-space sparsification:} a query-free black-box attack~\cite{songdemark2026} that projects the image toward a low-rank subspace within an $L_\infty$ constraint. Every attack is applied identically to all ten schemes at native resolution.

\smallskip
\noindent
\textbf{Metrics.} Our primary metric is the detection true-positive rate at 1\% false-positive rate, abbreviated \DetectAcc; it is a detection rate, not a classification accuracy. An image is declared watermarked when its matching-bit count reaches $\tau_K$, the smallest integer with $\Pr[\mathrm{Binom}(K,\tfrac12) \ge \tau_K]\le0.01$. This threshold is analytic under the unwatermarked null, so it is fixed by $K$ alone, consumes no calibration images, and is comparable across payloads. We also report bit accuracy (\BitAcc); since detection thresholds the same match count, the two need not move together. For visual quality we report PSNR, SSIM, and LPIPS~\cite{zhang2018lpips} between the watermarked and cover images, averaged over the test set.

\smallskip
\noindent
\textbf{Training protocol.} We train \ours at $128 \times 128$ with $K = 30$ bits, using Adam at learning rate $10^{-3}$, batch size 24, and 60 epochs. Image-quality losses are disabled for the first 3 epochs and ramped in over the next 6; distortion augmentation ($p_d = 0.6$) starts at epoch 8. SA-AT fires with probability $p_s = 0.5$, sampling rank $r \in [4,32]$ and running $T = 3$ inner steps at budget 0.05, the rank-$r$ basis refreshed every 200 steps. Every number comes from a single training run per configuration, so we report no seed variance. All models are evaluated on one NVIDIA RTX~6000 Ada GPU.

\subsection{Main Results}

\begin{table*}[t]
\centering
\resizebox{\textwidth}{!}{%
\begin{tabular}{ll cccccccccccccccc}
\toprule
\multirow{2}{*}{Dataset} & \multirow{2}{*}{Method} & \multicolumn{2}{c}{Clean} & \multicolumn{2}{c}{Sparsif.} & \multicolumn{2}{c}{Regen.} & \multicolumn{2}{c}{JPEG} & \multicolumn{2}{c}{Noise} & \multicolumn{2}{c}{Blur} & \multicolumn{2}{c}{Crop} & \multicolumn{2}{c}{Adv.} \\
\cmidrule(lr){3-4}\cmidrule(lr){5-6}\cmidrule(lr){7-8}\cmidrule(lr){9-10}\cmidrule(lr){11-12}\cmidrule(lr){13-14}\cmidrule(lr){15-16}\cmidrule(lr){17-18}
 & & BA$\uparrow$ & DA$\uparrow$ & BA$\uparrow$ & DA$\uparrow$ & BA$\uparrow$ & DA$\uparrow$ & BA$\uparrow$ & DA$\uparrow$ & BA$\uparrow$ & DA$\uparrow$ & BA$\uparrow$ & DA$\uparrow$ & BA$\uparrow$ & DA$\uparrow$ & BA$\uparrow$ & DA$\uparrow$ \\
\midrule
\multirow{10}{*}{COCO}
 & HiDDeN & 0.83 & 0.98 & 0.82 & 0.95 & 0.69 & 0.40 & 0.67 & 0.33 & 0.82 & 0.96 & 0.56 & 0.06 & 0.78 & 0.83 & 0.49 & 0.01 \\
 & MBRS & 0.82 & 0.94 & 0.80 & 0.94 & 0.69 & 0.33 & 0.70 & 0.40 & 0.81 & 0.94 & 0.57 & 0.04 & 0.78 & 0.84 & 0.48 & 0.01 \\
 & StegaStamp & 0.74 & 0.64 & 0.75 & 0.66 & 0.74 & 0.61 & 0.74 & 0.65 & 0.74 & 0.67 & 0.67 & 0.23 & 0.74 & 0.65 & 0.52 & 0.02 \\
 & SepMark & 0.84 & 0.99 & 0.81 & 0.94 & 0.66 & 0.22 & 0.63 & 0.15 & 0.83 & 0.99 & 0.53 & 0.01 & 0.74 & 0.61 & 0.48 & 0.01 \\
 & CIN & 0.77 & 0.63 & 0.71 & 0.45 & 0.52 & 0.02 & 0.52 & 0.00 & 0.68 & 0.37 & 0.49 & 0.01 & 0.51 & 0.01 & 0.51 & 0.00 \\
 & TrustMark & 1.00 & 1.00 & 0.66 & 0.33 & 0.49 & 0.00 & 0.97 & 0.94 & 1.00 & 1.00 & 0.98 & 0.96 & 0.99 & 0.99 & 0.50 & 0.01 \\
 & RivaGAN & 0.99 & 1.00 & 0.96 & 0.99 & 0.67 & 0.32 & 0.98 & 1.00 & 0.99 & 1.00 & 0.99 & 1.00 & 0.96 & 0.98 & 0.41 & 0.00 \\
 & WAM & 0.89 & 0.86 & 0.65 & 0.20 & 0.50 & 0.00 & 0.55 & 0.03 & 0.83 & 0.82 & 0.83 & 0.82 & 0.52 & 0.00 & 0.16 & 0.00 \\
 & DwtDctSvd & 1.00 & 1.00 & 0.53 & 0.14 & 0.53 & 0.07 & 0.91 & 1.00 & 0.99 & 1.00 & 0.87 & 1.00 & 0.53 & 0.02 & -- & -- \\
 \rowcolor{rowgray} & \textbf{\ours} & 1.00 & 1.00 & 1.00 & 1.00 & 0.78 & 0.77 & 0.93 & 1.00 & 1.00 & 1.00 & 0.59 & 0.12 & 0.66 & 0.24 & 0.32 & 0.00 \\
\midrule
\multirow{10}{*}{DIV2K}
 & HiDDeN & 0.83 & 0.98 & 0.82 & 0.94 & 0.68 & 0.33 & 0.66 & 0.32 & 0.84 & 0.98 & 0.56 & 0.00 & 0.78 & 0.78 & 0.51 & 0.02 \\
 & MBRS & 0.81 & 0.96 & 0.79 & 0.87 & 0.69 & 0.40 & 0.70 & 0.45 & 0.81 & 0.91 & 0.58 & 0.03 & 0.77 & 0.76 & 0.51 & 0.03 \\
 & StegaStamp & 0.74 & 0.66 & 0.75 & 0.65 & 0.74 & 0.66 & 0.73 & 0.58 & 0.75 & 0.69 & 0.66 & 0.21 & 0.74 & 0.62 & 0.53 & 0.01 \\
 & SepMark & 0.83 & 0.99 & 0.82 & 0.97 & 0.67 & 0.30 & 0.61 & 0.09 & 0.84 & 0.99 & 0.54 & 0.06 & 0.75 & 0.70 & 0.50 & 0.00 \\
 & CIN & 0.74 & 0.64 & 0.69 & 0.42 & 0.51 & 0.01 & 0.50 & 0.00 & 0.69 & 0.39 & 0.52 & 0.01 & 0.49 & 0.02 & 0.50 & 0.01 \\
 & TrustMark & 1.00 & 1.00 & 0.61 & 0.22 & 0.51 & 0.01 & 0.92 & 0.84 & 0.98 & 0.97 & 0.93 & 0.87 & 0.99 & 0.99 & 0.50 & 0.02 \\
 & RivaGAN & 0.99 & 0.99 & 0.94 & 0.99 & 0.66 & 0.24 & 0.96 & 0.97 & 0.98 & 0.99 & 0.97 & 0.98 & 0.93 & 0.93 & 0.43 & 0.00 \\
 & WAM & 0.88 & 0.84 & 0.66 & 0.23 & 0.51 & 0.01 & 0.57 & 0.07 & 0.84 & 0.83 & 0.83 & 0.81 & 0.51 & 0.00 & 0.15 & 0.00 \\
 & DwtDctSvd & 1.00 & 1.00 & 0.54 & 0.11 & 0.54 & 0.09 & 0.91 & 0.99 & 0.99 & 0.99 & 0.86 & 0.97 & 0.52 & 0.04 & -- & -- \\
 \rowcolor{rowgray} & \textbf{\ours} & 1.00 & 1.00 & 1.00 & 1.00 & 0.73 & 0.60 & 0.89 & 1.00 & 1.00 & 1.00 & 0.58 & 0.05 & 0.67 & 0.31 & 0.33 & 0.00 \\
\bottomrule
\end{tabular}}
\caption{Robustness on COCO and DIV2K, reported as bit accuracy (\BitAcc) and detection TPR at 1\% FPR (\DetectAcc). Clean is the unattacked watermark, and the attack columns are latent-space sparsification (Sparsif.), regeneration (Regen.), JPEG, noise, blur, crop, and white-box adversarial removal (Adv.).}
\label{tab:main}
\end{table*}

\begin{table}[t]
\centering
\resizebox{\columnwidth}{!}{%
\begin{tabular}{lcccccc}
\toprule
 & \multicolumn{3}{c}{COCO} & \multicolumn{3}{c}{DIV2K} \\
\cmidrule(lr){2-4}\cmidrule(lr){5-7}
 & PSNR$\uparrow$ & SSIM$\uparrow$ & LPIPS$\downarrow$ & PSNR$\uparrow$ & SSIM$\uparrow$ & LPIPS$\downarrow$ \\
\midrule
HiDDeN & 24.1 & 0.641 & 0.102 & 24.1 & 0.703 & 0.091 \\
MBRS & 24.1 & 0.636 & 0.103 & 24.1 & 0.698 & 0.087 \\
StegaStamp & 24.3 & 0.679 & 0.152 & 24.0 & 0.723 & 0.142 \\
SepMark & 24.3 & 0.647 & 0.108 & 24.2 & 0.705 & 0.093 \\
CIN & 41.5 & 0.984 & 0.001 & 41.1 & 0.986 & 0.001 \\
TrustMark & 41.2 & 0.985 & 0.001 & 40.4 & 0.987 & 0.001 \\
RivaGAN & 40.1 & 0.980 & 0.024 & 40.1 & 0.985 & 0.020 \\
WAM & 49.7 & 0.999 & 0.001 & 49.1 & 0.999 & 0.001 \\
DwtDctSvd & 36.9 & 0.972 & 0.012 & 36.0 & 0.973 & 0.012 \\
\midrule
\rowcolor{rowgray} \textbf{\ours} & 30.1 & 0.874 & 0.014 & 29.4 & 0.891 & 0.012 \\
\bottomrule
\end{tabular}}
\caption{Comparing \ours with the nine schemes in visual quality by PSNR, SSIM, and LPIPS.}
\label{tab:quality}
\vspace*{-2ex}
\end{table}

Tables~\ref{tab:main} and~\ref{tab:quality} report robustness and visual quality across the four attack families. We walk through them in turn.

Under standard distortions \ours is competitive on the photometric cases and weak on the spatial ones: roughly 0.9 \BitAcc under JPEG and 1.00 under additive noise on both datasets, with \DetectAcc 1.00 in each case, but only 0.59 \BitAcc and 0.12 \DetectAcc under Gaussian blur and 0.66 and 0.24 under an 80\% crop-and-resize, where a fixed spatial code is attenuated by low-pass filtering and desynchronized by resampling. TrustMark and RivaGAN, which operate at twice the resolution, keep a clear advantage here.

Regeneration is where \ours separates most clearly from the baselines. A single pass erases the faint, high-PSNR schemes: \BitAcc falls to 0.49 for TrustMark, 0.50 for WAM, 0.53 for DwtDctSvd, and 0.67 for RivaGAN, with \DetectAcc at or near zero, and no baseline exceeds 0.74 \BitAcc. \ours is the most robust evaluated scheme here, at \BitAcc 0.78 and \DetectAcc 0.77 on COCO (0.73 and 0.60 on DIV2K).

Under latent-space sparsification \ours holds \BitAcc and \DetectAcc at 1.00 on both datasets, and the margin does not erode as the attacker raises the projection rank. The most exposed are again the faint, high-PSNR watermarks of TrustMark, DwtDctSvd, and WAM, whose \DetectAcc falls to 0.33, 0.14, and 0.20, while the re-trained baselines and RivaGAN hold their clean detection level. Combining the two, \ours is the only evaluated method achieving high detection under both the selected regeneration and sparsification settings.

Adversarial removal is the weakest case. Under an $\varepsilon = 0.10$ white-box attack \ours reaches \BitAcc 0.32 and \DetectAcc 0 at the reported 30\,dB setting, comparable to the gradient-accessible WAM (0.16). Raising the embedding gain to 24\,dB recovers \BitAcc 0.86 and \DetectAcc 0.98, so robustness here is set by the embedding energy rather than the receiver; the Discussion returns to what this implies for deployment.

These results are consistent with the footprint hypothesis of Figure~\ref{fig:prelim}, though suggestive rather than conclusive, since \ours also differs from the baselines in decoder and training procedure. The clearest available control is embedding strength: \ours embeds a fainter watermark than the re-trained baselines (30 versus 24\,dB, Table~\ref{tab:quality}) yet holds regeneration \DetectAcc at 0.77 where they fall to 0.22--0.61, so the gap is not explained by embedding energy alone. Table~\ref{tab:quality} also shows how far apart the quality points are, from 24.1 to 49.7\,dB with \ours at 30.1\,dB: more imperceptible than the re-trained baselines and markedly less so than all five deployed checkpoints. The robustness columns should be read against that quality difference.

\begin{figure*}[t]
\centering
\includegraphics[width=\linewidth]{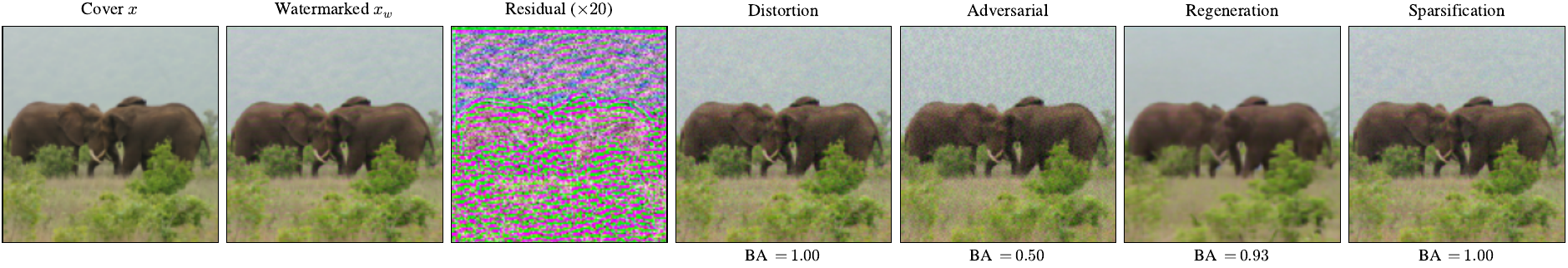}
\caption{Qualitative example for \ours: the cover image, the watermarked image $x_w$, the embedding residual scaled by $20\times$, and $x_w$ under the four attack families, each annotated with the \BitAcc recovered for this example.}
\label{fig:quality}
\end{figure*}

\subsection{Ablation Studies}

\begin{figure*}[t]
\centering
\includegraphics[width=\linewidth]{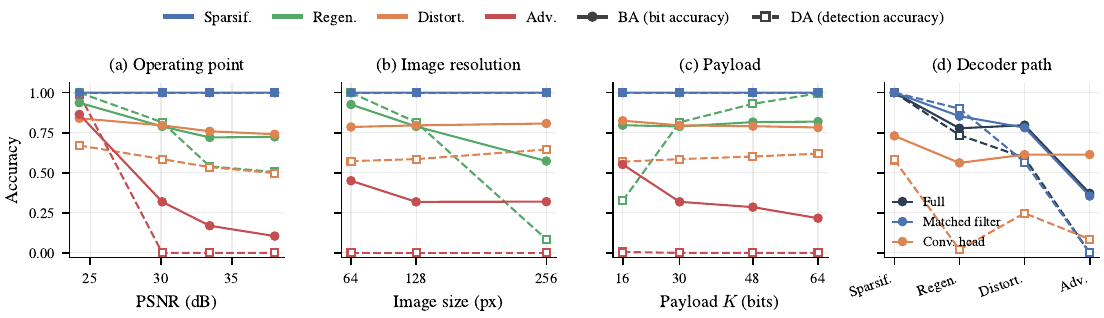}
\caption{Ablations for \ours on COCO under the four attack families. \textit{(a)~Operating Point:} embedding gain sweeping imperceptibility from 24 to 38\,dB PSNR. \textit{(b)~Image Resolution:} 64, 128, and 256 pixels. \textit{(c)~Watermark Payload:} message length from 16 to 64 bits. \textit{(d)~Decoder Path:} the full decoder versus the matched filter or the convolutional head alone.}
\label{fig:ablation}
\end{figure*}

We ablate \ours along four axes on COCO under the four attack families (Figure~\ref{fig:ablation}). These vary components of \ours and do not isolate the causal effect of spreading relative to the baselines, since \ours also changes the decoder and the training procedure at the same time.

\smallskip
\noindent
\textbf{Operating Point.}
The embedding gain trades imperceptibility for robustness. From 24 to 38\,dB, sparsification stays at 1.00 in \BitAcc and \DetectAcc, while the others weaken: regeneration \BitAcc falls from 0.94 to 0.72, distortion from 0.84 to 0.74, and adversarial removal from 0.86 to 0.10, with \DetectAcc following. Our 30\,dB setting is a deliberate bias toward imperceptibility, and is why adversarial removal is our weakest column.

\smallskip
\noindent
\textbf{Image Resolution.}
Our primary evaluation is at $128 \times 128$, and behavior is not uniform across scale. Across 64, 128, and 256 pixels, sparsification stays at 1.00 and distortion holds near 0.80 \BitAcc and 0.60 \DetectAcc, but regeneration degrades sharply: \BitAcc weakens from 0.93 to 0.57 and \DetectAcc collapses from 1.00 to near zero, as a fixed-length code spreads more thinly over a larger image. High-resolution regeneration is therefore a genuine limitation of the present design, and we do not claim this result transfers to full-resolution imagery.

\smallskip
\noindent
\textbf{Watermark Payload.}
Raising the payload from 16 to 64 bits leaves \BitAcc for sparsification, regeneration, and distortion essentially unchanged, while adversarial \BitAcc falls from 0.55 to 0.22, as each added bit divides a fixed energy budget. \DetectAcc mirrors this, except that a longer code sharpens the separation from the binomial null, raising regeneration \DetectAcc from 0.33 to 0.99 while its bit accuracy stays essentially flat.

\smallskip
\noindent
\textbf{Decoder Path.}
This ablation speaks most directly to the spreading mechanism, since the two paths share one trained model and differ only in which statistic is read out. The matched filter alone reproduces the full decoder under sparsification (1.00 \BitAcc), regeneration (0.85), and distortion (0.78), so the correlation receiver, not the convolutional head, supplies robustness in those three families. The head trails there but overtakes the matched filter under adversarial removal (0.61 versus 0.35 \BitAcc), which targets the correlation receiver directly, and is the fallback where the code desynchronizes. This localizes robustness within \ours but does not show that spreading is what distinguishes it from the baselines; the roles of embedding gain and SA-AT are given by the operating-point ablation and Table~\ref{tab:saat}.

\section{Discussion and Limitations}
\label{sec:disc}

\smallskip
\noindent
\textbf{Codebook Secrecy.}
\ours uses one secret codebook, fixed at training and reused for every image, and our threat model assumes it stays secret and is never statistically estimated. We do not test that assumption: we evaluate no known-message or chosen-message attack, no codeword estimation across images, no collusion, no watermark-copy attack, and no adaptive key recovery. A fixed dense code is a natural target for all of these; none of our results are evidence of security against them, and Theorem~\ref{thm:bound} does not apply, since its independence assumption is what such attacks violate. Dense spreading cuts both ways here: each marked image exposes more of the code than a concentrated residual would, so the property that resists oblivious removal may also ease code estimation. How many marked images such an estimator needs is the open question. Codebook rotation and multiple simultaneous keys are the most important directions for future work.

\smallskip
\noindent
\textbf{Adversarial Removal.}
The \DetectAcc of 0 that \ours records under white-box removal at 30\,dB is an energy mismatch, not a failure of the receiver: an $\varepsilon = 0.10$ perturbation can carry more energy than a 30\,dB watermark. The same decoder at the 24\,dB point recovers \BitAcc 0.86 and \DetectAcc 0.98, and Figure~\ref{fig:eps} shows it holding detection above 0.98 across the whole budget range while every other scheme, \ours at 30\,dB included, has lost detection by $\varepsilon = 0.06$. The attack therefore fixes where on the gain curve a deployment should sit, at a cost in visible quality, rather than bounding the design; it also presumes an attacker with the decoder gradients and hence the secret codewords, access a deployed detector does not expose~\citep{songdemark2026}.

\section{Conclusion}

\ours revisits spread-spectrum embedding in a neural post-hoc watermarking architecture, prompted by the measurement that existing encoder--decoder schemes concentrate each message bit in a small image-space footprint. It instead spreads every bit across the whole image as an explicit dense pseudo-random codeword, recovered by matched filtering with a convolutional fallback and sparsification-aware training. On COCO and DIV2K it uniquely retains high detection under both regeneration and latent-space sparsification, while remaining among the most imperceptible schemes. A conditional chip-space analysis explains why a dense code resists codeword-independent perturbations, though not the regeneration, compression, or code-estimating attacks we assess empirically; synchronization, scale, and codebook security remain future work.

\bibliography{refs}

\clearpage

\appendix

\section{Appendix}

\subsection{Implementation Details}
\label{app:impl}

Encoder and decoder share the HiDDeN backbone~\cite{zhu2018hidden} so that any robustness gain isolates to the spread-spectrum embedding and matched-filter receiver. The encoder applies four $3 \times 3$ convolutional blocks (Conv--BatchNorm--ReLU, 64 channels) to the cover image, concatenates the resulting features with the input image and the $K$ message bits broadcast as constant spatial maps, passes them through two further convolutional blocks, and maps the result to a three-channel residual bounded by a $\tanh$ of peak amplitude 0.15. When spreading is enabled the dense signal $\alpha\,s(b)$ of Eq.~(\ref{eq:embed}) is added to this residual, with $\alpha$ parameterized as the softplus of a free scalar initialized so that $\alpha = 0.06$. The decoder applies seven convolutional blocks (64 channels), global average pooling, and a linear head to $K$ bit logits; the spread path adds a $1 \times 1$ convolution to a single-channel chip map, the normalized correlation $\rho_i$ against the shared bank, a per-bit affine gain and bias, and a per-bit sigmoid gate whose bias is initialized to 2.0 so the matched filter is trusted at the start of training. The codewords are drawn once as i.i.d.\ Gaussian spatial maps with a fixed seed, normalized to unit per-pixel variance, and stored identically in the encoder and decoder.

\subsection{Proof of Theorem~\ref{thm:bound}}
\label{app:proof}

We restate the assumptions and give the argument in full. Theorem~\ref{thm:bound} is an idealized, conditional statement about the matched-filter path of the decoder; the assumptions below are modeling choices, and the Scope paragraph at the end of this section states what they exclude.

\smallskip
\noindent
\textbf{(A1) Codewords.} The codewords $\{c_i\}_{i=1}^{K}$, with $c_i\in\mathbb{R}^d$ and $d = HW$, have entries $c_{i,j}$ that are independent, zero-mean, and unit-variance after the per-pixel normalization $\tfrac1d\|c_i\|_2^2 = 1$, and are sub-Gaussian with variance proxy 1. This holds exactly for the Gaussian codewords we use and, by Hoeffding's lemma, for the symmetric $\pm1$ codewords examined in our codeword-design ablation (Table~\ref{tab:codeword}).

\smallskip
\noindent
\textbf{(A2) Idealized Residual.} The decoder reads bit $i$ from the matched-filter statistic $\rho_i=\langle\hat z,c_i\rangle/d$ on the recovered chip $\hat z=\alpha\,s(b)+\nu$, where $s(b)=\tfrac{1}{\sqrt K}\sum_{j}(2b_j-1)c_j$ and $\nu$ collects the cover-suppression residual and the cross-code interference. We idealize $\nu$ as carrying no systematic component along any secret codeword, that is $\mathbb{E}[\langle\nu,c_i\rangle]=0$ for every $i$. In the deployed system $\nu$ is produced by a learned projection of a real image and this is an approximation, not a property we establish. With (A1), (A2), and the near-orthogonality $\langle c_i,c_j\rangle/d\to 0$ for $i\neq j$, the clean statistic concentrates at margin $|\mathbb{E}\,\rho_i|=\alpha/\sqrt K$ with sign $2b_i-1$.

\smallskip
\noindent
\textbf{(A3) Codeword-independent Perturbation.} The chip-space perturbation $\Delta$ induced by the attack is statistically independent of the codebook $\{c_i\}$. This is the load-bearing assumption. It does not follow automatically for an attacker that observes $x_w$, because $x_w$ is itself a function of the codewords, so any attack that extracts code-correlated structure from the watermarked image falls outside the model.

\begin{proof}
Let the attack map $\hat z\mapsto\hat z+\Delta$ with $\Delta$ satisfying (A3) and $\|\Delta\|_2\le B$. The induced shift on bit $i$ is:
\begin{equation}
\Delta\rho_i \;=\; \frac{\langle\Delta,c_i\rangle}{d} \;=\; \frac1d\sum_{j=1}^{d}\Delta_j\,c_{i,j}.
\end{equation}
Because $\Delta$ is fixed independently of the codewords and each $c_{i,j}$ is zero-mean and unit-variance:
\begin{equation}
\mathbb{E}[\Delta\rho_i]=0,
\qquad
\mathrm{Var}[\Delta\rho_i]=\frac1{d^2}\sum_{j=1}^{d}\Delta_j^2
=\frac{\|\Delta\|_2^2}{d^2}\le\frac{B^2}{d^2}.
\end{equation}
As a fixed linear combination of independent sub-Gaussian variables, $\Delta\rho_i$ is itself sub-Gaussian with variance proxy $\sigma^2=\|\Delta\|_2^2/d^2\le B^2/d^2$. The decoder errs on bit $i$ only if the shift cancels the margin, that is if $\Delta\rho_i$ has sign opposite to $2b_i-1$ and magnitude at least $\alpha/\sqrt K$; the sub-Gaussian tail bound therefore gives
\begin{equation}
\begin{aligned}
\Pr[\widehat b_i\neq b_i]
&\;\le\;
\Pr \Big[|\Delta\rho_i|\ge\tfrac{\alpha}{\sqrt K}\Big] \\
&\;\le\;
\exp \Big(-\frac{(\alpha/\sqrt K)^2}{2\sigma^2}\Big)
\le
\exp \Big(-\frac{\alpha^2 d^2}{2KB^2}\Big),
\end{aligned}
\end{equation}
which is the stated per-bit bound. The counting step is a first-moment argument and yields an expected-budget statement, not a high-probability one: for the expected number of flips $\sum_{i}\Pr[\widehat b_i\neq b_i]\le K\exp(-\alpha^2 d^2/(2KB^2))$ to reach a constant fraction $\rho_0 K$ we need $\exp(-\alpha^2 d^2/(2KB^2))\ge\rho_0$, i.e.\ $\alpha^2 d^2/(2KB^2)\le\ln(1/\rho_0)$, which rearranges to
\begin{equation}
B\;\ge\;\frac{\alpha d}{\sqrt{2K\ln(1/\rho_0)}}\;=\;\Omega \Big(\frac{\alpha d}{\sqrt K}\Big).
\end{equation}
Finally, the attacker acts in pixel space through a perturbation $\delta$, and the chip map is $L$-Lipschitz, so $\|\Delta\|_2=\|\,\mathrm{chip}(x_w + \delta)-\mathrm{chip}(x_w)\|_2\le L\|\delta\|_2\le L\sqrt d\,\|\delta\|_\infty$. Requiring the chip budget $B=\Omega(\alpha d/\sqrt K)$ to be reachable under $\|\delta\|_\infty\le\varepsilon$ gives $L\sqrt d\,\varepsilon=\Omega(\alpha d/\sqrt K)$, that is
\begin{equation}
\varepsilon\;=\;\Omega \Big(\frac{\alpha\sqrt d}{L\sqrt K}\Big),
\end{equation}
so under (A1)--(A3) the pixel budget needed to disrupt matched-filter recovery of a dense code grows as $\sqrt d$.
\end{proof}

The $\sqrt d$ scaling is the formal content of the design principle: unable to align $\Delta$ with any codeword, the attack must spend energy spread over all $d$ coordinates. When instead each bit occupies only an effective fraction $\phi$ of the coordinates, its energy lies in a small distinctive subspace that a generic operation can suppress without the code, lowering the required budget by roughly $\sqrt\phi$.

\smallskip
\noindent
\textbf{Scope.} Three limitations follow directly from (A1)--(A3). First, the result is a statement about the matched-filter statistic under a codeword-independent perturbation; it is not a robustness proof for the deployed decoder, which also contains a convolutional head and a per-bit gate, and it does not establish robustness to regeneration, compression, denoising, or latent-space sparsification, all of which we treat as empirical questions. Second, it is an expected-budget argument: it lower-bounds the budget an attack needs for its expected flip count to reach a constant fraction, and says nothing about the tail of that count. Third, (A3) excludes exactly the adaptive threats that matter most for a fixed codebook, namely attacks that estimate the code from one or many watermarked images, known-message and chosen-message attacks, known-key attacks, and any perturbation whose statistics depend on the embedded code. We regard the theorem as an explanation of a design choice, not as a security guarantee.

\subsection{Robustness Versus Attack Strength}
\label{app:strength}

\begin{figure}[t]
\centering
\includegraphics[width=\columnwidth]{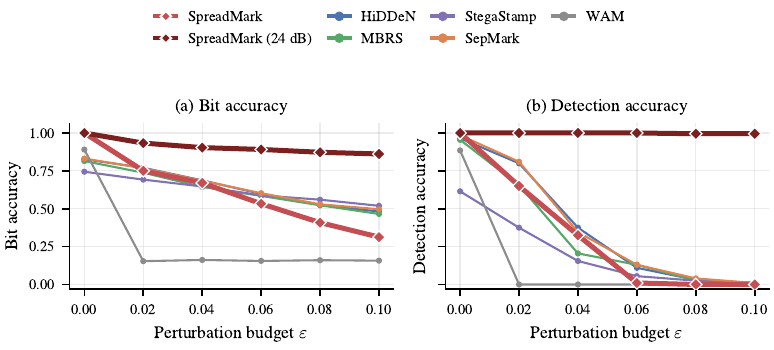}
\caption{White-box adversarial removal versus the perturbation budget $\varepsilon$ on COCO, as (a)~bit accuracy and (b)~detection accuracy, for \ours at 30 and 24\,dB and the baselines with a differentiable decoder.}
\label{fig:eps}
\end{figure}
\begin{figure*}[t]
\centering
\includegraphics[width=\linewidth]{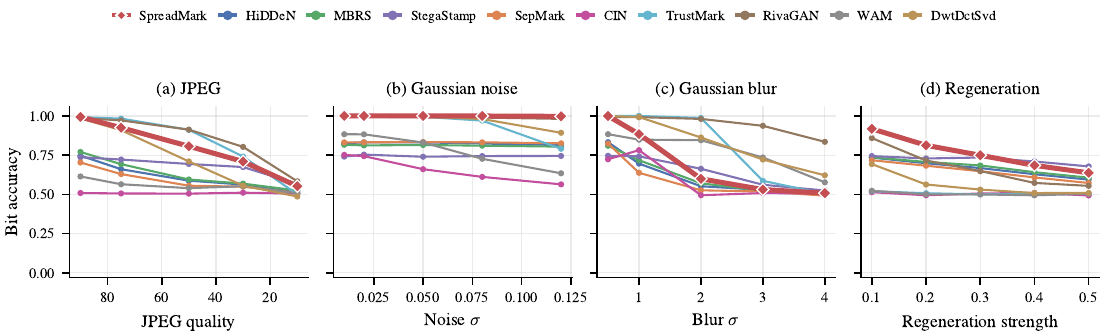}
\caption{Robustness versus attack strength on COCO for every scheme, reported as bit accuracy: (a)~JPEG quality, (b)~Gaussian noise, (c)~Gaussian blur, and (d)~Regeneration strength. \ours is drawn in bold.}
\label{fig:strength}
\end{figure*}

\noindent
\textbf{Adversarial Removal.} Figure~\ref{fig:eps} sweeps the white-box budget $\varepsilon$ for \ours at two quality settings and for the baselines that expose a differentiable decoder. At a matched budget the 24\,dB \ours holds detection at or above 0.98 across the whole range, whereas every other scheme, including the 30\,dB \ours, has lost detection by $\varepsilon = 0.06$; bit accuracy follows the same separation. The gap between the two \ours points, against an unchanged decoder, confirms that robustness in this regime is governed by the embedding gain and not by the receiver, consistent with the energy argument of Theorem~\ref{thm:bound}.

\smallskip
\noindent
\textbf{Distortion and Regeneration Strength.} Figure~\ref{fig:strength} sweeps the strength of each oblivious attack for all schemes at the 30\,dB quality setting. \ours retains full bit accuracy under additive noise across the tested range and degrades gracefully under JPEG compression. Blur is the sharp case: past kernel standard deviation 2 the low-pass filter desynchronizes the fixed spatial code, the geometric weakness discussed in the main text. Under regeneration \ours decays more slowly than every baseline at each strength.

\subsection{Effect of Sparsification-Aware Training}
\label{app:saat}

\begin{table}[t]
\centering
\small
\setlength{\tabcolsep}{9pt}
\renewcommand{\arraystretch}{1.15}
\begin{tabular}{lcccc}
\toprule
 & \multicolumn{2}{c}{\ours} & \multicolumn{2}{c}{w/o SA-AT} \\
\cmidrule(lr){2-3}\cmidrule(lr){4-5}
Attack & BA$\uparrow$ & DA$\uparrow$ & BA$\uparrow$ & DA$\uparrow$ \\
\midrule
Clean         & 1.00 & 1.00 & 1.00 & 1.00 \\
Sparsif.      & 1.00 & 1.00 & 1.00 & 1.00 \\
Regen         & \textbf{0.79} & \textbf{0.81} & 0.72 & 0.51 \\
JPEG          & \textbf{0.92} & \textbf{1.00} & 0.84 & 0.95 \\
Noise         & 1.00 & 1.00 & 1.00 & 1.00 \\
Blur          & 0.59 & 0.09 & 0.58 & 0.06 \\
Crop          & 0.67 & 0.25 & 0.69 & 0.40 \\
Adv.          & 0.32 & 0.00 & 0.31 & 0.00 \\
\midrule
PSNR (dB)     & \multicolumn{2}{c}{30.1} & \multicolumn{2}{c}{28.7} \\
\bottomrule
\end{tabular}
\caption{Effect of sparsification-aware adversarial training (SA-AT) at the 30\,dB operating point on COCO, as \BitAcc and \DetectAcc per attack for \ours with and without SA-AT.}
\label{tab:saat}
\end{table}

Table~\ref{tab:saat} separates the contribution of the SA-AT stage from that of the dense embedding, by retraining \ours with SA-AT disabled and everything else unchanged. The two properties we emphasize in the main text behave differently. Latent-space sparsification and clean recovery are already at 1.00 in both \BitAcc and \DetectAcc without SA-AT, so they are attributable to the embedding rather than to the training stage. Regeneration and JPEG are where SA-AT contributes, raising regeneration \DetectAcc from 0.51 to 0.81 and JPEG \DetectAcc from 0.95 to 1.00. Blur, crop, and adversarial removal are essentially unaffected. The comparison is not perfectly matched in quality, since the model without SA-AT converges to 28.7\,dB against 30.1\,dB for the full model, and the direction of that gap favors the ablated model; the regeneration difference is therefore, if anything, understated.

\subsection{Codeword Design}
\label{app:codeword}

\begin{table}[t]
\centering
\small
\resizebox{\columnwidth}{!}{%
\begin{tabular}{lcccccccc}
\toprule
Variant & Clean & Sparsif. & Regen & JPEG & Noise & Blur & Crop & Adv. \\
\midrule
\multicolumn{9}{l}{\textit{Bit accuracy (BA)}} \\
Gaussian, learnable $\alpha$ (\ours) & 1.00 & 1.00 & 0.79 & 0.92 & 1.00 & 0.59 & 0.67 & 0.32 \\
Bernoulli codewords & 1.00 & 1.00 & 0.70 & 0.77 & 1.00 & 0.57 & 0.69 & 0.32 \\
Fixed $\alpha$ & 1.00 & 1.00 & 0.86 & 0.94 & 1.00 & 0.61 & 0.68 & 0.35 \\
\midrule
\multicolumn{9}{l}{\textit{Detection accuracy (DA)}} \\
Gaussian, learnable $\alpha$ (\ours) & 1.00 & 1.00 & 0.81 & 1.00 & 1.00 & 0.09 & 0.25 & 0.00 \\
Bernoulli codewords & 1.00 & 1.00 & 0.44 & 0.73 & 1.00 & 0.05 & 0.37 & 0.00 \\
Fixed $\alpha$ & 1.00 & 1.00 & 0.94 & 1.00 & 1.00 & 0.14 & 0.36 & 0.00 \\
\bottomrule
\end{tabular}}
\caption{Codeword-design ablation near 30\,dB on COCO, as \BitAcc and \DetectAcc per attack, for Gaussian codewords with learned $\alpha$ (\ours), Bernoulli $\pm1$ codewords, and fixed $\alpha$.}
\label{tab:codeword}
\end{table}

Theorem~\ref{thm:bound} depends on the codewords only through their near-orthogonality, not through any particular distribution, and treats the embedding gain $\alpha$ as a fixed scalar. Table~\ref{tab:codeword} tests both assumptions by retraining \ours with two design changes at the 30\,dB quality setting: drawing the codewords from a symmetric $\pm1$ (Bernoulli) distribution instead of a Gaussian, and freezing $\alpha$ at the value the learned model converges to rather than training it. The properties the theorem governs are untouched: clean recovery and latent-space sparsification stay at 1.00 in both \BitAcc and \DetectAcc for every variant, and additive noise is likewise unaffected, which confirms that the sparsification robustness is a property of the dense code and not of how it is generated. Fixing $\alpha$ matches the learned model across the board and even improves regeneration slightly, so learning the amplitude end-to-end is a convenience rather than a necessity. Bernoulli codewords match on every column except regeneration, where they trail by a few points at a marginally lower quality setting. The spread-spectrum mechanism is thus robust to the codeword family and to whether the amplitude is learned, which supports reading Theorem~\ref{thm:bound} as a design principle that holds across reasonable codeword choices.

\end{document}